\documentclass[11pt]{article}
\usepackage[T1]{fontenc}
\usepackage{amsfonts}
\usepackage{amsmath}
\usepackage{amssymb}
\usepackage{amsthm}
\usepackage{bbm}
\usepackage{bm}
\usepackage{mathrsfs}
\usepackage{verbatim}
\usepackage{setspace}
\usepackage{color}
\usepackage{pdfsync}
\usepackage{enumitem}

\theoremstyle{plain}
\newtheorem{theorem}{Theorem}[section]

\newtheorem{lemma}[theorem]{Lemma}

\theoremstyle{definition}

\newtheorem{remark}[theorem]{Remark}

\newtheorem{example}[theorem]{Example}

\newtheorem{assumption}[theorem]{Assumption}
\theoremstyle{remark}

\renewenvironment{thebibliography}[1]{%
\begin{oldthebibliography}{#1}%
\setlength{\baselineskip}{.9em}
\linespread{1}
\small
\setlength{\parskip}{0.2ex}%
\setlength{\itemsep}{.40em}%
}%
{%
\end{oldthebibliography}%
}
\newcommand{\q}{\quad}

\newcommand{\eps}{\varepsilon}

\newcommand{\N}{\mathbb{N}}

\newcommand{\Q}{\mathbb{Q}}
\newcommand{\R}{\mathbb{R}}

\newcommand{\cB}{\mathcal{B}}

\newcommand{\cD}{\mathcal{D}}

\newcommand{\cF}{\mathcal{F}}

\newcommand{\cH}{\mathcal{H}}

\newcommand{\cP}{\mathcal{P}}

\newcommand{\fP}{\mathfrak{P}}

\DeclareMathOperator{\dom}{dom}
\DeclareMathOperator{\graph}{graph}

\DeclareMathOperator{\ri}{ri}
\DeclareMathOperator{\NA}{NA}

\DeclareMathOperator{\supp}{supp}
\DeclareMathOperator{\linspan}{span}

\newcommand{\qs}{\mbox{-q.s.}}

\newcommand{\1}{\mathbf{1}}
\newcommand{\sint}{\stackrel{\mbox{\tiny$\bullet$}}{}}

\newcommand{\DS}{\Delta S}

\numberwithin{equation}{section}

\usepackage[pdfborder={0 0 0}]{hyperref}
\hypersetup{
  urlcolor = black,
  pdfauthor = {Marcel Nutz},
  pdfkeywords = {Utility maximization; Knightian uncertainty; Nondominated model},
  pdftitle = {Utility Maximization under Model Uncertainty in Discrete Time},
  pdfsubject = {Utility Maximization under Model Uncertainty in Discrete Time},
  pdfpagemode = UseNone
}

\begin{document}

\title{\vspace{-.3em}
Utility Maximization under Model Uncertainty\\in Discrete Time
\date{\today}
\author{
  Marcel Nutz%
  \thanks{
  Department of Mathematics, Columbia University, mnutz@math.columbia.edu. Research supported by NSF Grant DMS-1208985.
  }
 }
}
\maketitle \vspace{-1em}

\begin{abstract}
We give a general formulation of the utility maximization problem under nondominated model uncertainty in discrete time and show that an optimal portfolio exists for any utility function that is bounded from above. In the unbounded case, integrability conditions are needed as nonexistence may arise even if the value function is finite.
\end{abstract}

\vspace{.9em}

{\small


\section{Introduction}\label{se:intro}

We study a robust utility maximization problem of the form
\begin{equation}\label{eq:eq:intro}
  u(x)=\sup_{H\in\cH_x} \inf_{P\in\cP} E_P[U(x+H\sint S_T)]
\end{equation}
in a discrete-time financial market. Here $S$ is the stock price process and $x+H\sint S_T$ is the agent's wealth at the time horizon $T$ resulting from the given initial capital $x\geq0$ and trading according to the portfolio $H$. Moreover, $U$ is a utility function defined on the positive half-line and $\cP$ is a set of probability measures. Hence, the agent attempts to choose a portfolio $H$ from the set $\cH_x$ of all admissible strategies such as to maximize the worst-case expected utility over the set $\cP$ of possible models. A distinct feature of our problem is that $\cP$ may be nondominated in the sense that there exists no reference probability measure with respect to which all $P\in\cP$ are absolutely continuous.

We show in a general setting that an optimal portfolio $\hat{H}\in\cH_x$ exists for any utility function that is bounded from above. In the classical theory where $\cP$ is a singleton, existence holds also in the unbounded case, under the condition that $u(x)<\infty$. This is not true in our setting; we exhibit counterexamples for any utility function unbounded from above. Positive results can be obtained under suitable integrability assumptions.

Our basic model for the financial market is the one proposed in \cite{BouchardNutz.13}; in particular, we use their no-arbitrage condition $\NA(\cP)$. It states that for any portfolio $H$ which is riskless in the sense that $H\sint S_T\geq0$ holds $\cP$-quasi-surely ($P$-a.s.\ for all $P\in\cP$), it follows that $H\sint S_T=0$ $\cP$-quasi-surely. While this is not the only possible choice of a viability condition (see e.g.~\cite{AcciaioBeiglbockPenknerSchachermayer.12, CoxObloj.11}), an important motivation for our study is to show that $\NA(\cP)$ implies the existence of optimal portfolios in a large class of utility maximization problems, which we consider a desirable feature for a no-arbitrage condition. Thus, to state a clear message in this direction, we aim at a general result in an abstract setting. (It will be clear from the examples that the possible nonexistence for unbounded $U$ has little to do with no-arbitrage considerations.)

The main difficulty in the nondominated case is the failure of the so-called Komlos-type arguments that are used extensively e.g.\ in \cite{DelbaenSchachermayer.94, KramkovSchachermayer.99}.
We shall instead use dynamic programming, basically along the lines of \cite{RasonyiStettner.05, RasonyiStettner.06} but of course without a reference measure, to reduce to the case of a one-period market with deterministic initial data. Such a market is still nondominated, but the set of portfolios is a subset of a Euclidean space where it is easy to obtain compactness from the no-arbitrage condition. The passage from the one-period markets to the original market is achieved via measurable selections and in particular the theory of lower semianalytic functions; it draws from \cite{BouchardNutz.13, Nutz.10Gexp, NutzVanHandel.12, SonerTouziZhang.2010dual}.

The utility maximization problem for a singleton $\cP$ has a long and rich history in mathematical finance; we refer to \cite[Section~2]{FollmerSchied.04} or~\cite[Section~3]{KaratzasShreve.98} for background and references. In the discrete-time case, the most general existence result was obtained in~\cite{RasonyiStettner.06}; it establishes the existence of an optimal portfolio under the standard no-arbitrage condition $\NA$ (which is the same as our $\NA(\cP)$ when $\cP$ is a singleton) for any concave nondecreasing function $U$, under the sole assumption that $u(x)<\infty$.
Robust or ``maxmin''-criteria as in~\eqref{eq:eq:intro} are classical in decision theory; the systematic analysis goes back at least to Wald (see the survey~\cite{Wald.50}). A solid axiomatic foundation was given in modern economics; a landmark paper in this respect is~\cite{GilboaSchmeidler.89}. Most of the literature on the robust utility maximization in mathematical finance, starting with~\cite{Quenez.04} and~\cite{Schied.06}, assumes that the set $\cP$ is dominated by a reference measure $P_*$; we refer to~\cite{Schied.06} for an extensive survey. While this assumption has no clear foundation from an economic or decision-theoretic point of view, it is mathematically convenient and in particular allows to work within standard probability theory. The nondominated problem is quite different in several respects; for instance, finding a worst-case measure does not lead to an optimal portfolio in general.

To the best of our knowledge, there are two previous existence results for optimal portfolios in the nondominated robust utility maximization problem, both in the context of price processes with continuous paths and restricted to specific utility functions (power, logarithm or exponential). In~\cite{TevzadzeToronjadzeUzunashvili.13}, a factor model is studied and the model uncertainty extends over a compact (deterministic) set of possible drift and volatility coefficients, which leads to a Markovian problem that is tackled by solving an associated Hamilton--Jacobi--Bellman--Isaacs partial differential equation under suitable conditions. In~\cite{MatoussiPossamaiZhou.12utility}, the uncertainty is over a compact set of volatility coefficients whereas the drift is known, but possibly non-Markovian. The problem is tackled by solving an associated second order backward stochastic differential equation under suitable conditions. We mention that in both cases, the involved set of measures is relatively compact and the volatilities are supposed to be uniformly nondegenerate, which acts as an implicit no-arbitrage condition.
In~\cite{DenisKervarec.13}, the authors consider a more general form of uncertainty about drift and volatility over a compact set of measures and a general (bounded) utility function. They establish a minimax result and the existence of a worst-case measure under the assumption that each $P\in\cP$ admits an equivalent martingale measure. Turning to different but related nondominated problems in the mathematical finance literature, \cite{BayraktarHuang.13} studies robust maximization of asymptotic growth under covariance uncertainty, \cite{FernholzKaratzas.11} analyzes optimal arbitrage under model uncertainty and \cite{Cont.06} and \cite{TalayZheng.02} consider risk management under model uncertainty. The robust superhedging problem has been studied in several settings; see \cite{AcciaioBeiglbockPenknerSchachermayer.12, BeiglbockHenryLaborderePenkner.11, BouchardNutz.13, DenisMartini.06, DolinskySoner.12, GalichonHenryLabordereTouzi.11, NeufeldNutz.12, Peng.10,SonerTouziZhang.2010rep, SonerTouziZhang.2010dual}, among others.

The remainder of this note is organized as follows. In Section~\ref{se:onePeriod}, we study in detail the case of a one-period market, establish existence under an integrability condition for $U^+$, and discuss how nonexistence can arise when $U$ is unbounded from above. In Section~\ref{se:multiPeriod}, consider the multi-period case and mainly focus on the case where $U$ is bounded from above (Theorem~\ref{th:multiperiodBdd}), as this seems to be the only case allowing for a general theory without implicit assumptions. The case of unbounded $U^+$ is discussed in a restricted setting (Example~\ref{ex:bddS}).

\section{The One-Period Case}\label{se:onePeriod}

Let $(\Omega,\cF)$ be a measurable space. We consider a stock price process given by a deterministic vector $S_0\in\R^d$ and an $\cF$-measurable, $\R^d$-valued random vector $S_1$; we write $\Delta S$ for $S_1-S_0$. In this setting, a portfolio is a deterministic vector $h\in\R^d$ and the corresponding gain from trading is given by the inner product $h\DS=\sum_{i=1}^d h^{i} \Delta S^i$.
We are given a nonempty convex set $\cP$ of probability measures on $\cF$. A subset $A\subseteq \Omega$ is called $\cP$-polar if $A\subseteq A'$ for some $A'\in\cF$ satisfying $P(A')=0$ for all $P\in\cP$, and a property is said to hold $\cP$-quasi surely or $\cP$-q.s.\ if it holds outside a $\cP$-polar set. Given some initial capital $x\geq0$, the set of admissible portfolios is defined by
\[
  D_x:=\{h\in\R^d:\, x+h\Delta S \geq 0 \; \cP\qs\}.
\]
We shall work under the no-arbitrage condition $\NA(\cP)$ of \cite{BouchardNutz.13}; that is, given $h\in\R^d$,
\begin{equation}\label{eq:NA1}
  h\DS \geq 0 \quad\cP\qs \q\q\mbox{implies}\q\q h\DS = 0 \quad\cP\qs
\end{equation}
As a preparation for the analysis of the multi-period case, we consider a utility function $U$ that can be random. The following condition is in force throughout this section.

\begin{assumption}\label{as:U1}
  The function $U: \Omega\times [0,\infty)\to[-\infty,\infty)$ is such that
  \begin{enumerate}
    \item $\omega\mapsto U(\omega,x)$ is $\cF$-measurable and bounded from below for each $x>0$,
    \item $x\mapsto U(\omega,x)$ is nondecreasing, concave and continuous for each $\omega\in\Omega$.
  \end{enumerate}
\end{assumption}

In particular, $U$ is finite-valued on $\Omega\times (0,\infty)$, while $U(\omega,0)=\lim_{x\downarrow 0}U(\omega,x)$ may be infinite. We shall sometimes omit the first argument and write $U(x)$ for $U(\omega,x)$. Moreover, we define  $U(x):=-\infty$ for $x<0$ and note that $U$ is $\cF\otimes \cB(\R)$-measurable as a consequence of a standard result on Carath\'eodory functions \cite[Lemma~4.51, p.\,153]{AliprantisBorder.06}.
We can now state the main result of this section.

\begin{theorem}\label{th:existence1underInt}
  Let $\NA(\cP)$ hold and $x\geq0$. Assume that
  \begin{equation}\label{eq:int1}
    E_P[U^+(x+h\Delta S)]<\infty\quad \mbox{for all}\quad h\in D_x\mbox{ and }P\in\cP.
  \end{equation}
  Then
  \[
    u(x):=\sup_{h\in D_x} \inf_{P\in\cP} E_P[U(x+h\Delta S)] < \infty
  \]
  and there exists $\hat{h}\in D_x$ such that
  $
    \inf_{P\in\cP} E_P[U(x+\hat{h}\Delta S)] =  u(x).
  $
\end{theorem}

The proof is stated below. Due to the infimum over $\cP$, the optimal portfolio $\hat{h}$ is \emph{not unique} in general, even if $U$ is strictly concave and there are no redundant assets.
The following counterexample shows that the integrability condition~\eqref{eq:int1} cannot be dropped: existence of an optimal portfolio may fail even if $u(x)<\infty$.

\begin{example}\label{ex:nonexistence2d}
  Let $d=2$ and let $U$ be a deterministic function which is strictly concave and \emph{unbounded} from above. Moreover, let $x>0$, $\Omega=\R^2$, $S_0\equiv (1,1)$ and $\Delta S (\omega) = \omega$ for all $\omega\in\R^2$. Let $P_1$ be the probability measure on $\R^2$ such that $\DS^1$ and $\DS^2$ are independent and
  \begin{align*}
    P_1\{\DS^1=-1\}&=  P_1\{\DS^1=1\}=1/2, \\
    P_1\{\DS^2=-1\}&= P_1\{\DS^2=2\}=1/2.
  \end{align*}
  Moreover, let $P_2$ be a second probability such that $\DS^1$ and $\DS^2$ are independent under $P_2$ and
  \begin{align*}
      &P_2\{\DS^1=-1\}=1/2, \quad P_2\{\DS^1\geq0\}=1/2,\quad E_{P_2}[U^+((\DS^1)^+)]=\infty, \\
      &P_2\{\DS^2=0\}=1.
  \end{align*}
  Let $\cP$ be the convex hull of $\{P_1,P_2\}$. Then $\NA(\cP)$ holds and $u(x)<\infty$, but there exists no optimal portfolio $\hat{h}\in D_x$.
\end{example}

\begin{proof}
  We consider the case $U(0)>-\infty$; the case $U(0)=-\infty$ is similar with minor modifications. It is elementary to check that $\NA(\cP)$ holds. We observe that $S^1$ is a martingale under $P_1$, whereas $S^2$ has a strictly positive rate of return. It then follows that the problem
  \[
     u_1(x):=\sup_{h\in\R^2}E_{P_1}[U(x+h\DS)]
  \]
  admits a unique optimal portfolio $\hat{g}$ of the form $\hat{g}=(0,\hat{g}^2)$ for some $\hat{g}^2>0$. Moreover, we have $u(x)\leq u_1(x)<\infty$.
  Suppose that $\hat{h}=(\hat{h}^1,\hat{h}^2)\in D_x$ is an optimal portfolio for $\cP$. As $\DS^1$ is unbounded from above under $P_2$, $\hat{h}\in D_x$ implies $\hat{h}^1\geq0$. Moreover, using that $E_{P_2}[U(x+h\DS)]=\infty$ whenever $h=(h^1,h^2)\in D_x$ satisfies $h^1>0$,
  we see that
  \[
    \inf_{P\in\cP} E_P[U(x+h\DS)]=
    \begin{cases}
      E_{P_1}[U(x+h\DS)] & \text{if } h^1>0, \\
      E_{P_1}[U(x+h^2\DS^2)]\wedge U(x) & \text{if } h^1=0.
    \end{cases}
  \]
  Since $E_{P_1}[U(x+h\DS)]>U(x)$ if $h^1>0$ is close to zero and $h^2=\hat{g}^2$, we must have $\hat{h}^1>0$. However, a direct calculation shows that the function
  \[
    h^1\mapsto E_{P_1}[U(x+h^1\DS^1 + \hat{h}^2\DS^2)]
  \]
  is strictly decreasing in $h_1\geq0$, so that the maximum cannot be attained at a strictly positive number.
\end{proof}

\begin{remark}
  In the notation of Example~\ref{ex:nonexistence2d} and its proof, define the random utility function $\tilde{U}(\omega,x):=U(x+ \hat{g}^2\Delta S^2(\omega))$ and consider only $S^1$ as a tradable asset. Then the above arguments again imply the nonexistence of an optimal portfolio.
\end{remark}

The previous discussion leaves open the case where $U$ is deterministic and $d=1$. Somewhat curiously, we have the following positive result; the proof is stated at the end of this section.

\begin{remark}\label{rk:1dimExistence}
  Assume that $d=1$ and that $U$ is deterministic with $U(0)>-\infty$. If $\NA(\cP)$ holds, $x\geq0$ and $u(x)<\infty$, then there exists $\hat{h}\in D_x$ such that\footnote{Note that all the involved expectations are well defined due to $U(0)>-\infty$.}
  $
    \inf_{P\in\cP} E_P[U(x+\hat{h}\Delta S)] = u(x).
  $
\end{remark}

\subsection{Proofs}
Let us now turn to the proofs of Theorem~\ref{th:existence1underInt} and Remark~\ref{rk:1dimExistence}.
We fix the initial capital $x\geq0$; moreover, $\NA(\cP)$ is always in force. Some additional notation is needed. The quasi-sure support $\supp_{\cP} (\Delta S)$ is defined as the smallest closed subset $A$ of $\R^d$ such that $P\{\Delta S \in A\}=1$ for all $P\in\cP$; cf.\ \cite{BouchardNutz.13}. We can then introduce
\[
  L:=\linspan \supp_{\cP} (\Delta S) \subseteq \R^d,
\]
the smallest linear subspace of $\R^d$ containing $\supp_{\cP} (\Delta S)$. In view $\NA(\cP)$, the fundamental theorem of asset pricing in the form of \cite[Theorem~3.1]{BouchardNutz.13} shows that the origin is in the closed convex hull of $\supp_{\cP} (\Delta S)$ and thus $L$ coincides with the affine hull of $\supp_{\cP} (\Delta S)$. The orthogonal complement
\[
  L^\bot:=\{h\in\R^d:\, hv=0\mbox{ for all }v\in L\}
\]
is the nullspace of $\DS$ in the sense of the following lemma, which entails that projecting onto $L$ eliminates any redundancy between portfolios.

\begin{lemma}\label{le:orthogonal}
  Given $h\in\R^d$, we have $h\in L^\bot$ if and only if $h\Delta S=0$ $\cP$-q.s.
\end{lemma}

\begin{proof}
  If $h\in L^\bot$, then $P\{h\DS=0\}=P\{\DS \in L\}=1$ for all $P\in\cP$ and hence $h\Delta S=0$ $\cP$-q.s.
  Conversely, let $h\notin L^\bot$; then there exists $v\in \supp_{\cP}(\DS)$ such that $hv\neq0$, and thus $hv'\neq0$ for all $v'$ in an open neighborhood $B(v)$ of $v$. By the minimality property of the support, it follows that there exists $P\in \cP$ such that $P\{\DS\in B(v)\} >0$. Therefore, $P\{h\DS\neq0\}>0$ and $h$ does not satisfy $h\Delta S=0$ $\cP$-q.s.
\end{proof}

The following compactness property is an important consequence of $\NA(\cP)$.

\begin{lemma}\label{le:admCompact}
  The set $K_x:=D_x \cap L \subseteq \R^d$ is convex, compact and contains the origin.
\end{lemma}

\begin{proof}
  It is clear that $K_x$ is convex, closed and contains the origin. Suppose for contradiction that $K_x$ is unbounded; then there are $h_n\in K_x$ such that $|h_n|\to\infty$. After passing to a subsequence, $h_n/|h_n|$ converges to a limit $h\in\R^d$. As $K_x$ is convex, we have $h_n/|h_n|\in K_x$, thus $h\in K_x$ by the closedness. Moreover, $|h|=1$. Since $h_n\Delta S \geq -x$ $\cP$-q.s.\ for all $n$, we see that
  $h\Delta S=\lim h_n \Delta S /|h_n| \geq 0$ $\cP$-q.s., which implies $h\Delta S=0$ $\cP$-q.s.\ by $\NA(\cP)$. As $h\in K_x\subseteq L$, it follows that $h=0$, contradicting that $|h|=1$.
\end{proof}

\begin{lemma}\label{le:upperBound}
  Let~\eqref{eq:int1} hold. Then there exists a random variable $Y\geq0$ satisfying $E_P[Y]<\infty$ for all $P\in\cP$ and
  \[
    U^+(x+h\Delta S) \leq Y\quad\cP\qs \quad \mbox{for all}\quad h\in D_x.
  \]
\end{lemma}

\begin{proof}
  The following arguments are quite similar to~\cite{RasonyiStettner.06}. As $h\Delta S=0$ $\cP$-q.s.\ for $h\in D_0$, the claim is clear for $x=0$; we suppose that $x>0$.
  By projecting onto $L$, it suffices to consider $h\in K_x$. Let $g_1,\dots,g_N\in\R^d$ be such that the convex cone generated by $g_1,\dots,g_N$ equals $\linspan K_x$, where $N$ is chosen minimally. Let $h\in K_x$; then $h=\sum_{i=1}^N \lambda_i g_i$ for some $\lambda_i\geq0$, and as $K_x$ is bounded by Lemma~\ref{le:admCompact}, there exists a constant $c\geq1$ independent of $h$ such that $|\lambda_i|\leq c/N$.
  As a result,
  \[
    x+h\Delta S = x+ \sum_{i=1}^N \lambda_i g_i\Delta S \leq x+ c\max \{0,g_1\Delta S,\dots, g_N\Delta S\}.
  \]
  Let
  \[
    Y:=U^+(x+ c\max \{0,g_1\Delta S,\dots, g_N\Delta S\})
  \]
  and fix $P\in\cP$. To show that $E_P[Y]<\infty$, it suffices to establish that
  \[
    E_P[U^+(x+ c g_i\Delta S)]<\infty
  \]
  for each $i$. To this end, fix an arbitrary $g\in \ri (K_x)$ and let $\eps\in(0,1)$ be such that $\tilde{g}_i:=g + \eps(cg_i-g)\in K_x$. In view of Assumption~\ref{as:U1}(i), by adding a constant to $U$, we may suppose that $U(1)\geq0$, and then we have the elementary inequality
  \begin{equation}\label{eq:scalingIneq}
    \eps U^+(y) \leq 2U^+(\eps y) + 2U(2),\quad y\in\R;
  \end{equation}
  see~\cite[Lemma~2]{RasonyiStettner.06}.
  Therefore,
  \begin{align*}
    \eps U^+(x+ c g_i\Delta S)
    &= \eps U^+(x+ g\Delta S+[c g_i-g]\Delta S) \\
    & \leq 2U^+(\eps [x+ g\Delta S]+\eps [c g_i-g]\Delta S) + 2U(2)\\
    & \leq 2U^+(x+ g\Delta S+\eps [c g_i-g]\Delta S) + 2U(2)\\
    & = 2U^+(x+ \tilde{g}_i\Delta S) + 2U(2)
  \end{align*}
  holds $\cP$-q.s.; namely, on the set $\{x+ g\Delta S\geq0\}$. The first term above is $P$-integrable by~\eqref{eq:int1}. The same holds for the second term; in fact, \eqref{eq:int1} immediately implies that $E_P[U^+(x)]<\infty$, and using the concavity of $U$ and $x>0$, it follows that $E_P[U^+(y)]<\infty$ for all $y\geq0$.
\end{proof}

\begin{proof}[Proof of Theorem~\ref{th:existence1underInt}]
  Lemma~\ref{le:upperBound} and Fatou's lemma imply that for all $P\in\cP$, the function $h\mapsto E_P[U(x+h\Delta S)]$ is upper semicontinuous on $D_x$. It follows that $h\mapsto\inf_{P\in\cP} E_P[U(x+h\Delta S)]$ is upper semicontinuous and thus attains its (finite) supremum on the compact set $K_x$ (Lemma~\ref{le:admCompact}). Finally, \[\sup_{h\in K_x}\inf_{P\in\cP} E_P[U(x+h\Delta S)] = \sup_{h\in D_x}\inf_{P\in\cP} E_P[U(x+h\Delta S)]\] by Lemma~\ref{le:orthogonal}.
\end{proof}

Next, we state two auxiliary results that will be used in the analysis of the multi-period case.

\begin{lemma}\label{le:countableSup}
  Let~\eqref{eq:int1} hold and let $\cD\subseteq D_x$ be dense. Then
  \[
    \sup_{h\in \cD} \inf_{P\in\cP} E_P[U(x+h\Delta S)] = u(x).
  \]
\end{lemma}

\begin{proof}
  The function $h\mapsto \phi(h):=\inf_{P\in\cP} E_P[U(x+h\Delta S)]$ is a concave on $D_x$. Thus, $\phi$ is continuous on the relative interior of its domain $\dom(\phi)$. Moreover, $\sup_{\dom(\phi)} \phi=\sup_{\ri(\dom(\phi))}\phi$ by concavity, so that it suffices to show that $\cD$ is dense in $\dom(\phi)$.
  Indeed, suppose that $x>0$ and let $h\in\ri(D_x)$. As $0\in D_x$, there exists $\lambda>1$ such that $\lambda h\in D_x$; that is, $x+h\DS \geq (1-\lambda^{-1})x$ $\cP$-q.s. In view of Assumption~\ref{as:U1}(i), this implies that $U(x+h\Delta S)$ is uniformly bounded from below and thus that $\phi(h)>-\infty$. On the other hand, if $x=0$, then $h\in D_x$ implies $h\DS=0$ $\cP$-q.s.\ by $\NA(\cP)$. Thus, we have $\ri(D_x)\subseteq \dom(\phi)$ in both cases, and hence $\cD$ is dense in $\dom(\phi)$.
\end{proof}

\begin{lemma}\label{le:valueFctContinuous}
  Let~\eqref{eq:int1} hold for some $x>0$. Then~\eqref{eq:int1} holds for all $x\geq0$ and $u: [0,\infty)\to [-\infty,\infty)$ is nondecreasing, concave and continuous.
\end{lemma}

\begin{proof}
 The first claim follows from~\eqref{eq:scalingIneq} and implies that $u(x)<\infty$ for $x\in[0,\infty)$. Moreover, it is elementary to see that $u$ is nondecreasing and concave. In particular, $u$ is continuous on $(0,\infty)$, the interior of its domain. It remains to show that $u(0)\geq \lim_{n\to\infty} u(1/n)$. For each $n\geq1$, let $\hat{h}_n\in K_{1/n}$ be an optimal portfolio for $x=1/n$ as in Theorem~\ref{th:existence1underInt}. Since $K_{1/n}\subseteq K_1$ and $K_1$ is compact, we have $\hat{h}_n\to h_\infty$ after passing to a subsequence. Moreover, we have $h_\infty \in \cap_{n\geq1} K_{1/n}=K_0$. (In fact, $h_\infty=0$ and the whole sequence converges.)
  Using Lemma~\ref{le:upperBound} (with $x=1$) and Fatou's lemma similarly as in the proof of Theorem~\ref{th:existence1underInt}, we obtain that
  \[
    \lim_{n\to\infty} u(1/n)
     = \lim_{n\to\infty} \inf_{P\in\cP} E_P[U(1/n+\hat{h}_n\DS)]
     \leq \inf_{P\in\cP} E_P[U(0+h_\infty\DS)] \leq u(0)
  \]
  as desired.
\end{proof}

It remains to show the result for the scalar case, Remark~\ref{rk:1dimExistence}.

\begin{proof}[Proof of Remark~\ref{rk:1dimExistence}]
  Note that all expectations are well defined because $U$ is bounded from below on $[0,\infty)$.
  As in the proof of Theorem~\ref{th:existence1underInt}, we may assume that $x>0$; moreover, there exists a maximizing sequence $h_n\in K_x$ converging to some $\hat{h}\in K_x$.
  By passing to a subsequence, we may assume that one of the following holds:
   \begin{enumerate}
     \item $h_n>0$ for all $n\geq1$ and $\{h_n\}$ is a monotone sequence,
     \item $h_n<0$ for all $n\geq1$ and $\{h_n\}$ is a monotone sequence,
     \item $h_n=0$ for all $n\geq1$.
   \end{enumerate}
   Case (iii) is trivial while (i) and (ii) are symmetric; we focus on (i).
  As $u(x)<\infty$ and $h_n>0$, the set
  \[
    \cP_*:=\{P\in\cP:\, E_P[U^+(x+h\Delta S)]<\infty \mbox{ for some }h>0\}
  \]
  is not empty. Using again~\eqref{eq:scalingIneq}, we see that in fact
  \begin{equation}\label{eq:cPstarAllh}
    \cP_*= \{P\in\cP:\, E_P[U^+(x+h\Delta S)]<\infty \mbox{ for all }h\geq 0\}.
  \end{equation}
  We have $h_n\in[\hat{h},h_1]$ if $\{h_n\}$ is decreasing, or otherwise $h_n\in[h_1,\hat{h}]$. Hence,
  \[
    U(x+h_n\Delta S) \leq U^+(x+h_1\Delta S) \vee U^+(x+\hat{h}\Delta S),\quad n\geq1.
  \]
  For $P\in\cP_*$, the right-hand side is integrable by~\eqref{eq:cPstarAllh} and so Fatou's lemma yields that $\limsup_{n\to\infty} E_P[U(x+h_n\Delta S)] \leq E_P[U(x+\hat{h}\Delta S)]$. As a result,
  \begin{align}
  \inf_{P\in\cP_*} E_P[U(x+\hat{h}\Delta S)]
  & \geq \limsup_{n\to\infty} \inf_{P\in\cP_*} E_P[U(x+h_n\Delta S)] \nonumber \\
  &\geq \limsup_{n\to\infty} \inf_{P\in\cP} E_P[U(x+h_n\Delta S)] \nonumber \\
  &= u(x) \label{eq:proof1dimLimit}.
  \end{align}
  Recalling (i), we clearly have $\hat{h}\geq0$. If $\hat{h}>0$, then $E_P[U(x+\hat{h}\Delta S)]=\infty$ for $P\in\cP\setminus\cP_*$ and thus
  \begin{equation}\label{eq:cPstarEqual}
  \inf_{P\in\cP} E_P[U(x+\hat{h}\Delta S)] =\inf_{P\in\cP_*} E_P[U(x+\hat{h}\Delta S)].
  \end{equation}
  If $\hat{h}=0$, then~\eqref{eq:cPstarEqual} is still true since both sides are equal to $U(x)$; recall that $U$ is deterministic.
  In view of~\eqref{eq:proof1dimLimit}, this completes the proof.
\end{proof}

\section{The Multi-Period Case}\label{se:multiPeriod}

Let us now detail the setting for the multi-period market; we follow~\cite{BouchardNutz.13}. Fix a time horizon $T\in\N$, let $\Omega_1$ be a Polish space and let $\Omega_t:=\Omega_1^t$ be the $t$-fold Cartesian
product, $t=0,1,\dots,T$, with the convention that $\Omega_0$ is a singleton. We define $\cF_t$ to be the universal completion of the Borel-$\sigma$-field $\cB(\Omega_t)$; that is,
$\cF_t = \cap_{P} \cB(\Omega_t)^P$, where $\cB(\Omega_t)^P$ is the $P$-completion of $\cB(\Omega_t)$ and $P$ ranges over the set $\fP(\Omega_t)$ of all probability measures on $\cB(\Omega_t)$.
Moreover, we set $(\Omega,\cF):=(\Omega_T,\cF_T)$; this will be our basic measurable space. For convenience of notation, we shall often see $(\Omega_t,\cF_t)$ as a subspace of $(\Omega,\cF)$.

For each $t\in\{0,1,\dots,T-1\}$ and $\omega\in\Omega_t$, we are given a nonempty convex set
$\cP_t(\omega)\subseteq \fP(\Omega_1)$; intuitively,
$\cP_t(\omega)$ is the set of possible models for the $t$-th period, given state $\omega$ at time $t$. We assume that for each $t$,
\[%
  \graph(\cP_t):=\{(\omega,P):\, \omega\in \Omega_t,\, P\in\cP_t(\omega)\}\subseteq \Omega_t \times \fP(\Omega_t)\quad \mbox{is analytic,}
\]%
where we use the usual weak topology on $\fP(\Omega_1)$.
We recall that a subset of a Polish space is called analytic if it is the image of a Borel subset of another Polish space under a Borel-measurable mapping (see \cite[Chapter~7]{BertsekasShreve.78}); in particular, the above condition is satisfied whenever $\graph(\cP_t)$ is a Borel set. Analyticity of $\graph(\cP_t)$ implies that $\cP_t$ admits a universally measurable selector; that is, a universally measurable kernel $P_t:\, \Omega_t\to \fP(\Omega_1)$ such that $P_t(\omega)\in \cP_t(\omega)$ for all $\omega\in\Omega_t$. Given such a kernel $P_t$ for each $t\in\{0,1,\dots T-1\}$, we can define a probability $P$ on $\Omega$ by Fubini's theorem,
\[
  P(A)=\int_{\Omega_1} \cdots \int_{\Omega_1} \1_A(\omega_1,\dots,\omega_T) P_{T-1}(\omega_1,\dots,\omega_{T-1}; d\omega_T)\cdots P_0(d\omega_1),\quad A\in \Omega,
\]
where we write $\omega=(\omega_1,\dots,\omega_T)$ for a generic element of $\Omega$. The above formula will be
abbreviated as $P=P_0\otimes \cdots \otimes P_{T-1}$ in what follows. We can then introduce the set $\cP\subseteq\fP(\Omega)$ of possible models for the multi-period market up to time $T$,
\[%
  \cP:=\{P_0\otimes \cdots \otimes P_{T-1}:\, P_t(\cdot)\in \cP_t(\cdot),\,t=0,1,\dots T-1\},
\]%
where, more precisely, each $P_t$ is a universally measurable selector of $\cP_t$. See also~\cite{BouchardNutz.13, DolinskyNutzSoner.11} for more background and examples.

Next, we introduce the stocks and trading portfolios. Let $d\in \N$ and let $S_t=(S_t^1,\dots,S_t^d): \Omega_t \to \R^d$ be Borel-measurable for all $t\in \{0,1,\dots,T\}$.
We assume that $S^i_t\geq0$ $\cP$-q.s. Moreover, let $\cH$ be the set of all predictable $\R^d$-valued processes. Given $H\in\cH$, we denote
\[%
  H\sint S=(H\sint S_t)_{t\in\{0,1,\dots,T\}},\quad H\sint S_t=\sum_{u=1}^t H_u \Delta S_u,
\]%
where $\Delta S_u=S_u-S_{u-1}$. Sometimes it will be convenient to write $\DS_{t+1}$ explicitly as a function on $\Omega_t\times \Omega_1$,
\[
  \Delta S_{t+1}(\omega,\omega') = \Delta S_{t+1}(\omega_1,\dots,\omega_t,\omega'),\quad (\omega,\omega')=((\omega_1,\dots,\omega_t),\omega')\in \Omega_t\times \Omega_1.
\]
For fixed initial capital $x\geq0$, the set of admissible portfolios for our utility maximization problem is given by
\[
  \cH_x:=\{H\in\cH:\, x+ H\sint S_t\geq0\; \cP\qs\mbox{ for }t=1,\dots,T\}.
\]
We continue to work under the no-arbitrage condition $\NA(\cP)$ from~\cite{BouchardNutz.13}; in the present setting, it postulates that for all $H\in\cH$,
\[
  H\sint S_T \geq 0 \quad\cP\qs \q\q\mbox{implies}\q\q H\sint S_T = 0 \quad\cP\qs
\]
Recall that a function $f$ from a Borel subset of a Polish space into $\overline{\R}:=[-\infty,\infty]$ is called lower semianalytic if $\{f<c\}$ is analytic for all $c\in\R$; in particular, any Borel function is lower semianalytic.

We have seen in the previous section that nonexistence of an optimal strategy may arise when the utility function is unbounded from above. While we have used the integrability condition~\eqref{eq:int1} in the one-period case (which already is not sharp), it seems difficult to find a condition in the multi-period case that is actually verifiable (or at least sharp). In fact, even the finiteness of the value function in the case without uncertainty is typically difficult to verify. In order to establish a clean statement, we therefore focus on the bounded case for our main result, and give an example for the unbounded case below.

\begin{theorem}\label{th:multiperiodBdd}
  Let $\NA(\cP)$ hold, let $x\geq0$ and let $U : \Omega\times [0,\infty)\to\R$ be a lower semianalytic function which is bounded from above and satisfies Assumption~\ref{as:U1}.
  Then there exists $\hat{H}\in\cH_x$ such that
  \[
    \inf_{P\in\cP} E_P[U(x+\hat{H}\sint S_T)] =  \sup_{H\in \cH_x} \inf_{P\in\cP} E_P[U(x+H \sint S_T)]<\infty.
  \]
\end{theorem}

The proof is stated below. Our last result is a simple example where the utility maximization problem admits a solution for any deterministic utility function, possibly unbounded from above, under a strong boundedness and nondegeneracy assumption on the stock price process. Let us first explain what we mean by nondegeneracy. Given $t\in \{0,\dots,T-1\}$, $\omega\in\Omega_t$ and $x\geq0$, let $D_{t,x}(\omega):=\{h\in\R^d:\, x+h\DS_{t+1}(\omega,\cdot) \geq 0\; \cP_t(\omega)\qs\}$ and let $K_{t,x}(\omega)$ be the corresponding projection as in Lemma~\ref{le:admCompact}. If $\NA(\cP)$ holds, then by Lemma~\ref{le:admCompact} and Lemma~\ref{le:NAlocal} below there exists a function $\eps: \Omega_t\to \R$ which is strictly positive $\cP$-q.s.\ and has the following property: for all $h\in K_{t,x}(\omega)$ with $|h|=1$, there exists $P\in\cP_t(\omega)$ such that $P\{h\Delta S_{t+1}(\omega,\cdot)<-\eps(\omega)\}>0$. We shall say that $S$ is \emph{uniformly nondegenerate} if $\eps$ can be chosen to be a positive constant. (In the spirit of~\cite{Schal.00}, one could also call this a uniform no-arbitrage condition.) We then have the following result, again proved in the next subsection.

\begin{example}\label{ex:bddS}
   Let $x\geq0$ and let $U : \Omega\times [0,\infty)\to\R$ be a lower semianalytic function such that $U(\cdot,1)$ is bounded from above and Assumption~\ref{as:U1} holds. Suppose that $S$ is bounded and uniformly nondegenerate. Then there exists $\hat{H}\in\cH_x$ such that
  \[
    \inf_{P\in\cP} E_P[U(x+\hat{H}\Delta S)] =  \sup_{H\in \cH_x} \inf_{P\in\cP} E_P[U(x+H \sint S_T)]<\infty.
  \]
\end{example}

\subsection{Proofs}

For fixed $\omega\in\Omega_t$, the random variable $\Delta S_{t+1}(\omega,\cdot)$ on $\Omega_1$ determines a one-period market on $(\Omega_1,\cF_1)$ under the set $\cP_t(\omega)\subseteq \fP(\Omega_1)$. We denote the no-arbitrage condition~\eqref{eq:NA1} of that market by $\NA(\cP_t(\omega))$. The following lemma, proved in \cite[Theorem~4.5]{BouchardNutz.13}, will be useful in order to apply the results of Section~\ref{se:onePeriod}.

\begin{lemma}\label{le:NAlocal}
  The following are equivalent:
  \begin{enumerate}
    \item $\NA(\cP)$ holds.
    \item The set $\{\omega\in\Omega_t:\, \NA(\cP_t(\omega)) \mbox{ fails}\}$ is $\cP$-polar for all $t\in\{0,\dots, T-1\}$.
  \end{enumerate}
\end{lemma}

In view of the discussion on possible no-arbitrage conditions mentioned in the Introduction, let us remark that the above ``locality property'' of $\NA(\cP)$ is crucial in order to apply dynamic programming as in the subsequent arguments.

We also need a local description of the admissible portfolios.

\begin{lemma}\label{le:Hxpointwise}
  Let $x\geq0$ and $H\in\cH$. Then $H\in\cH_x$ if and only if
  \[
    x+H\sint S_{t+1}(\omega,\cdot)\geq0\;\cP_t(\omega)\qs\quad \mbox{for $\cP$-quasi-every }\omega\in\Omega_t
  \]
  and every $t=0,1,\dots,T-1$.
\end{lemma}

\begin{proof}
  The ``if'' implication is a direct consequence of Fubini's theorem; we show the converse. Let $t\in\{0,\dots, T-1\}$, let $H\in\cH_x$ and set
  \[
    B:=\big\{\omega\in \Omega_t:\, \mbox{$\{x+H\sint S_{t+1}(\omega,\cdot)\geq0\}$ is not $\cP_t(\omega)$-polar}\big\};
  \]
  then we need to prove that $B$ is $\cP$-polar.
  It follows from \cite[Proposition~7.29, p.\,144]{BertsekasShreve.78} that the mapping
  \[
    \R\times \R^d\times \Omega_t\times \fP(\omega_1)\to \overline{\R},\quad (x,h,\omega,P)\mapsto E_P[(x+h\DS_{t+1}(\omega,\cdot))^-]
  \]
  is Borel-measurable. As the graph of $\cP_t$ is analytic, this implies that the set-valued mapping
  \[
    \Psi(x,h,\omega):=\{P\in\cP_t(\omega): E_P[(x+h\DS_{t+1}(\omega,\cdot))^-]>0\}
  \]
  has an analytic graph. Using the Jankov--von Neumann Theorem \cite[Proposition~7.49, p.\,182]{BertsekasShreve.78}, we can then find a universally measurable mapping
  \[
    P_t: \R\times \R^d\times \Omega_t\to \fP(\Omega_1)
  \]
  such that $P_t(x,h,\omega)\in\cP_t(\omega)$ for all $x,h,\omega$ and $P_t(x,h,\omega)\in \Psi(x,h,\omega)$ on $\{\Psi\neq\emptyset\}$.
  Since compositions of universally measurable mappings remain universally measurable \cite[Proposition~7.44, p.\,172]{BertsekasShreve.78}, the kernel
  \[
    \omega\mapsto P'_t(\omega):=P_t(x+H\sint S_t(\omega), H_t(\omega),\omega)
  \]
  is again universally measurable. We then have
  \[
    B=\big\{\omega\in \Omega_t:\, P'_t(\omega)\{x+H\sint S_{t+1}(\omega,\cdot)<0\}>0\big\},
  \]
  showing that $B$ is universally measurable. Suppose that there exists $P\in\cP$ such that $P(B)>0$, then if $P':=P\otimes_t P'_t\in\fP(\Omega_{t+1})$ is the product measure formed from $P'_t$ and the restriction of $P$ to $\Omega_t$, we have $P'\{x+ H\sint S_{t+1}<0\}>0$, contradicting that $H\in\cH_x$ and $P'\in\cP$. Therefore, $B$ is $\cP$-polar.
\end{proof}

The conditions of Theorem~\ref{th:multiperiodBdd} are in force throughout the remainder of its proof.
In order to employ dynamic programming, we introduce the conditional value functions at the intermediate times. There are some measure-theoretic issues related to the simultaneous presence of suprema and infima in our problem, so we shall work with certain regular versions of the value functions.
Set $U(x):=-\infty$ for $x<0$ and denote $U_T:=U$. If $\omega\in\Omega_t$ and $\omega'\in\Omega_1$, we also write $\omega\otimes_t\omega'$ for $(\omega,\omega')\in\Omega_{t+1}$. For $t=T-1,\dots,0$ and $\omega\in\Omega_t$, define
\begin{align}
  U_t(\omega,x) &:= \sup_{h\in \Q^d} \inf_{P\in\cP_t(\omega)} E_P[U_{t+1}(\omega\otimes_t\cdot, x+ h\DS_{t+1}(\omega,\cdot))],\quad x>0,\label{eq:defUt}\\
  U_t(\omega,0) &:= \lim_{x\downarrow 0} U_t(\omega,x)\label{eq:defUtzero}
\end{align}
as well as $U_t(\omega,x)=-\infty$ for $x<0$. The following lemma ensures that $U_t$ is well defined for all $t$.

\begin{lemma}\label{le:UisLSA}
  Let $t\in\{0,\dots,T\}$. Then $U_t: \Omega_t\times [0,\infty) \to[-\infty,\infty)$ is lower semianalytic, bounded from above and satisfies Assumption~\ref{as:U1}.
\end{lemma}

\begin{proof}
  The claim is true by assumption for $t=T$; we show the induction step from $t+1$ to $t$.
  It is elementary to see that $U_t(\omega,\cdot)$ is nondecreasing and bounded from above. Let $y>0$; then by assumption there exists $c\in\R$ such that $U_{t+1}(\cdot,y)\geq c$. Considering $h=0$ in~\eqref{eq:defUt}, it follows that
  \[
    U_t(\omega,x)\geq \inf_{P\in\cP_t(\omega)} E_P[U_{t+1}(\omega\otimes_t\cdot, x)] \geq c.
  \]
  Next, we show the concavity on $(0,\infty)$. Let $x_1,x_2\in(0,\infty)$, let $\eps>0$ and let $h_1,h_2\in\Q^d$ be such that
  \[
    U_t(\omega,x_i) \leq \eps + \inf_{P\in\cP_t(\omega)} E_P[U_{t+1}(\omega\otimes_t\cdot, x_i+ h_i\DS_{t+1}(\omega,\cdot))],\quad i=1,2.
  \]
  Using the concavity of $U_{t+1}$ and $(h_1+h_2)/2\in\Q^d$, we see that
  \begin{align*}
  &\frac{U_t(\omega,x_1)+ U_t(\omega,x_2)}{2} \\
  &\leq \eps + \inf_{P\in\cP_t(\omega)} E_P\left[U_{t+1}\left(\omega\otimes_t\cdot, \frac{x_1+x_2}{2}+ \frac{h_1+h_2}{2}\DS_{t+1}(\omega,\cdot)\right)\right]\\
  &\leq \eps + \sup_{h\in\Q^d} \inf_{P\in\cP_t(\omega)} E_P\left[U_{t+1}\left(\omega\otimes_t\cdot, \frac{x_1+x_2}{2}+ h\DS_{t+1}(\omega,\cdot)\right)\right].
  \end{align*}
  As $\eps>0$ was arbitrary, it follows that $[U_t(\omega,x_1)+ U_t(\omega,x_2)]/2 \leq U_t(\omega,[x_1+x_2]/2)$; i.e., $U_t(\omega,\cdot)$ is midpoint-concave on $(0,\infty)$. Since moreover $U_t(\omega,\cdot)$ is bounded on any closed subinterval of $(0,\infty)$, this implies that $U_t(\omega,\cdot)$ is indeed concave on $(0,\infty)$; see, e.g., \cite[p.\,12]{Donoghue.69}. In particular, $U_t(\omega,\cdot)$ is continuous on $(0,\infty)$. In view of the definition~\eqref{eq:defUtzero}, both concavity and continuity extend to $[0,\infty)$.

  It remains to show that $U_t$ is lower semianalytic. Since the precomposition of a lower semianalytic function with a Borel mapping is again lower semianalytic \cite[Lemma~7.30(3), p.\,177]{BertsekasShreve.78}, we see that the function
  \[
    \Omega_t\times \Omega_1 \times (0,\infty)\times \R^d \to \overline{\R},\quad (\omega,\omega',x,h) \mapsto U_{t+1}(\omega\otimes_t\omega', x+ h\DS_{t+1}(\omega,\omega'))
  \]
  is lower semianalytic. Using a fact about Borel kernels acting on lower semianalytic functions \cite[Proposition~7.48, p.\,180]{BertsekasShreve.78}, we can deduce that
  \[
    \Omega_t\times \fP(\Omega_1) \times (0,\infty)\times \R^d \to \overline{\R},\quad (\omega,P,x,h) \mapsto E_P[U_{t+1}(\omega\otimes_t\cdot, x+ h\DS_{t+1}(\omega,\cdot))]
  \]
  is also lower semianalytic. Since the graph of $\cP_t$ is analytic, it then follows by \cite[Proposition~7.47, p.\,179]{BertsekasShreve.78} that
  \[
    \phi: \Omega_t\times (0,\infty)\times \R^d \to \overline{\R},\quad \phi(\omega,x,h):=\inf_{P\in\cP_t(\omega)} E_P[U_{t+1}(\omega\otimes_t\cdot, x+ h\DS_{t+1}(\omega,\cdot))]
  \]
  is lower semianalytic. Finally, we have $U_t(\omega,x)= \sup_{h\in \Q^d} \phi(\omega,x,h)$ by definition, and the supremum of a \emph{countable} family of lower semianalytic functions is still lower semianalytic \cite[Lemma~7.30(2), p.\,177]{BertsekasShreve.78}. Thus, $U_t$ is lower semianalytic as a function on $\Omega_t\times (0,\infty)$. Finally, $U_t: \Omega_t\times [0,\infty)\to \overline{\R}$ is also lower semianalytic; indeed, for each $c\in\R$, it follows from~\eqref{eq:defUtzero} that $\{(\omega,x)\in\Omega_t\times [0,\infty):  U_t(\omega,x)<c\}$ is the countable union of the analytic sets $\{(\omega,x)\in\Omega_t\times (0,\infty):  U_t(\omega,x)<c\}$ and $\{(\omega,x)\in\Omega_t\times [0,\infty):  x=0,\, U_t(\omega,1/n)<c\}$, $n\geq1$.
\end{proof}

While the nonstandard definitions~\eqref{eq:defUt} and~\eqref{eq:defUtzero} were made to ensure the crucial measurability properties of $U_t$, the following lemma shows that $U_t$ is indeed a version of the conditional value function at time $t$.

\begin{lemma}\label{le:supOverQsame}
  Let $t\in \{0,\dots,T-1\}$. For $\cP$-quasi-every $\omega\in\Omega_t$,
  \begin{equation}\label{eq:supOverQsame}
  U_t(\omega,x)
  = \sup_{h\in D_{t,x}(\omega)} \inf_{P\in\cP_t(\omega)} E_P[U_{t+1}(\omega\otimes_t\cdot, x+ h\DS_{t+1}(\omega,\cdot))], \quad x\geq0,
  \end{equation}
  where $D_{t,x}(\omega):=\{h\in\R^d:\, x+h\DS_{t+1}(\omega,\cdot) \geq 0\; \cP_t(\omega)\qs\}$.
\end{lemma}

\begin{proof}
  Let $\omega\in\Omega_t$ and $x\geq0$. If $h\in\Q^d\setminus D_{t,x}(\omega)$, then for some $P\in\cP_t(\omega)$, we have $P\{x+ h\DS_{t+1}(\omega,\cdot)<0\}>0$ and hence $E_P[U_{t+1}(\omega\otimes_t\cdot, x+ h\DS_{t+1}(\omega,\cdot))]=-\infty$. This implies the inequality ``$\leq$'' of~\eqref{eq:supOverQsame}. To see the reverse inequality, we first focus on $x>0$. As $S_{t+1}$ is nonnegative $\cP$-q.s., we have $\Delta S_{t+1}(\omega,\cdot)\geq -S_t(\omega)$ $\cP_t(\omega)$-q.s.\ for all $\omega$ outside a $\cP$-polar set (argue as in the proof of Lemma~\ref{le:Hxpointwise}). It follows that any $h=(h^1,\dots,h^d)\in[0,\infty)^d$ with $|h^1+\cdots+h^d|\leq x/|S_t(\omega)|$ is contained in the convex set $D_{t,x}(\omega)$ and thus that $\Q^d$ is dense in $D_{t,x}(\omega)$. If $\omega$ is also such that $\NA(\cP_t(\omega))$ is satisfied, then the claimed inequality follows by Lemma~\ref{le:countableSup}, for any $x>0$. To extend the result to $x=0$, it then suffices to note that both sides of~\eqref{eq:supOverQsame} are continuous in $x\in[0,\infty)$: this holds for $U_t(\omega,\cdot)$ by the definition~\eqref{eq:defUtzero}, and for the right-hand side by Lemma~\ref{le:valueFctContinuous}. It remains to recall Lemma~\ref{le:NAlocal}(ii).
\end{proof}

\begin{lemma}\label{le:optimalhSelector}
  Let $t\in\{0,\dots,T-1\}$, $x\geq0$ and $H\in\cH_x$. There exists a universally measurable mapping $\hat{h}_t: \Omega_t\to \R^d$ such that $x+ H\sint S_t(\omega)+ \hat{h}_t(\omega)\DS_{t+1}(\omega,\cdot)\geq0$ $\cP_t(\omega)$-q.s.\ and
  \begin{equation}\label{eq:optimalhSelectorDef}
   \inf_{P\in\cP_t(\omega)} E_P[U_{t+1}(\omega\otimes_t\cdot, x+ H\sint S_t(\omega)+ \hat{h}_t(\omega)\DS_{t+1}(\omega,\cdot))] = U_t(\omega,x+ H\sint S_t(\omega))
  \end{equation}
  for $\cP$-quasi-every $\omega\in\Omega_t$.%
\end{lemma}

\begin{proof}
  We first show that $U_t$ is $\cF_t\otimes \cB(\R)$ measurable; recall that $\cF_t$ is the universal $\sigma$-field of $\Omega_t$. Indeed, we know from Lemma~\ref{le:UisLSA} that $U_t(\cdot,x)$ is $\cF_t$-measurable for fixed $x\in\R$ and that $U_t(\omega,\cdot)$ is continuous on $[0,\infty)$ for fixed $\omega\in\Omega_t$. Thus, $U$ is product-measurable as a function on $\Omega_t\times [0,\infty)$. As $U_t(\cdot,x)=-\infty$ for $x<0$, it follows that $U_t$ is product-measurable as a function on $\Omega_t\times \R$ as claimed.
  Next, we show that the function
  \[
    \phi(\omega,x,h):=\inf_{P\in\cP_t(\omega)} E_P[U_{t+1}(\omega\otimes_t\cdot, x+h\DS_{t+1}(\omega,\cdot))]
  \]
  is $\cF_t\otimes\cB(\R)\otimes \cB(\R^d)$-measurable.
  Indeed, for fixed $(x,h)\in[0,\infty)\times\R^d$, it follows from the proof of Lemma~\ref{le:UisLSA} that $\omega\mapsto \phi(\omega,x,h)$ is lower semianalytic and in particular universally measurable. On the other hand, $U_{t+1}(\tilde{\omega},\cdot)$ is upper semicontinuous for any $\tilde{\omega}\in\Omega_{t+1}$. Since $U_{t+1}$ is bounded from above, an application of Fatou's lemma yields that
  $(x,h)\mapsto \phi(\omega,x,h)$ is upper semicontinuous for every $\omega\in\Omega_t$. It now follows as in~\cite[Lemma~4.12]{BouchardNutz.13} that
  $\phi$ is $\cF_t\otimes\cB(\R)\otimes \cB(\R^d)$-measurable.

  Fix $x\geq0$ and consider the set-valued mapping
  \[
   \Phi(\omega):=\{h\in\R^d:\, \phi(\omega,x+ H\sint S_t(\omega),h) = U_t(\omega,x+ H\sint S_t(\omega)) \},\quad\omega\in\Omega_t.
  \]
  By the above, its graph is in $\cF_t\otimes \cB(\R^d)$. Thus, $\Phi$ admits an $\cF_t$-measurable selector $\tilde{h}_t$ on the universally measurable set $\{\Phi\neq\emptyset\}$; cf.\ the corollary and scholium of \cite[Theorem~5.5]{Leese.75}. We extend $\tilde{h}_t$ by setting $\tilde{h}_t:=0$ on $\{\Phi\neq\emptyset\}$. Moreover, Theorem~\ref{th:existence1underInt} and Lemma~\ref{le:supOverQsame} show that  $\Phi(\omega)\neq\emptyset$ for $\cP$-quasi-every $\omega\in\Omega_t$, so that $\tilde{h}_t$ solves~\eqref{eq:optimalhSelectorDef} $\cP$-q.s.

  Let $B$ be the set  of all $\omega\in\Omega_t$ where we do not have $x+ H\sint S_t(\omega)+ \tilde{h}_t(\omega)\DS_{t+1}(\omega,\cdot)\geq0$ $\cP_t(\omega)$-q.s. As in the proof of Lemma~\ref{le:Hxpointwise}, $B$ is universally measurable, so that $\hat{h}_t:=\tilde{h}_t\1_{B^c}$ is still universally measurable. Moreover, as
  \[
    \inf_{P\in\cP_t(\omega)} E_P[U_{t+1}(\omega\otimes_t\cdot, x+ H\sint S_t(\omega)+ \tilde{h}_t(\omega)\DS_{t+1}(\omega,\cdot))]=-\infty,\quad \omega\in B,
  \]
  $\hat{h}_t$ is still satisfies~\eqref{eq:optimalhSelectorDef} $\cP$-q.s.\ and the proof is complete.
\end{proof}

\begin{lemma}\label{le:epsOptimalP}
  Let $t\in\{0,1,\dots,T-1\}$, $H\in\cH_x$ and
  \[
    I_t(\omega):= \inf_{P\in\cP_t(\omega)} E_P[U_{t+1}(\omega\otimes_t\cdot, x+ H\sint S_{t+1}(\omega,\cdot))],\quad \omega\in\Omega_t.
  \]
  Given $\eps>0$, there exists a universally measurable kernel $P_t^\eps: \Omega_t\to \fP(\Omega_1)$ such that $P_t^\eps(\omega)\in\cP_t(\omega)$ for all $\omega\in\Omega_t$ and
  \[
    E_{P_t^\eps(\omega)}[U_{t+1}(\omega\otimes_t\cdot, x+ H\sint S_{t+1}(\omega,\cdot))] \leq
  \begin{cases}
    I_t(\omega) + \eps & \text{if } I_t(\omega)>-\infty, \\
    -\eps^{-1} & \text{if } I_t(\omega)=-\infty.
  \end{cases}
  \]
\end{lemma}

\begin{proof}
  The function $(\omega,P,x,h)\mapsto E_P[U_{t+1}(\omega\otimes_t\cdot, x+ h\DS_{t+1}(\omega,\cdot))]$ is lower semianalytic by the proof of Lemma~\ref{le:UisLSA}. As the graph of $\cP_t$ is analytic, the Jankov--von Neumann Theorem in the form of \cite[Proposition~7.50, p.\,184]{BertsekasShreve.78} implies that there exists a universally measurable mapping $(\omega,x,h)\mapsto \tilde{P}_t^\eps(\omega,x,h)\in\cP_t(\omega)$ such that
  \[
    E_{\tilde{P}_t^\eps(\omega,x,h)}[U_{t+1}(\omega\otimes_t\cdot, x+ h\DS_{t+1}(\omega,\cdot))] \leq
  \begin{cases}
    I_t(\omega) + \eps & \text{if } I_t(\omega)>-\infty, \\
    -\eps^{-1} & \text{if } I_t(\omega)=-\infty.
  \end{cases}
  \]
  The composition $\omega\mapsto P_t^\eps(\omega):=\tilde{P}_t^\eps(\omega,x+H\sint S_t(\omega),H_{t+1}(\omega))$ has the desired properties.
\end{proof}

\begin{proof}[Proof of Theorem~\ref{th:multiperiodBdd}.]
  Let $\hat{h}_0$ be an optimal portfolio for $\inf_{P\in\cP_0}E_P[U_1(x+h\DS_1)]$ as in Lemma~\ref{le:optimalhSelector}, and set $\hat{H}_1:=\hat{h}_0$. Proceeding recursively, use Lemma~\ref{le:optimalhSelector} to define the strategy $\omega\mapsto \hat{h}_t(\omega)$ for $\inf_{P\in\cP_t(\omega)} E_P[U_{t+1}(\omega\otimes_t\cdot, x+\hat{H}\sint S_t + h\DS_{t+1}(\omega,\cdot))]$ and set $\hat{H}_{t+1}:=h_t$, for $t=1,\dots,T-1$. We then have $\hat{H}\in\cH_x$ by Lemma~\ref{le:Hxpointwise}. To establish that $\hat{H}$ is optimal, we first show that
  \begin{equation}\label{eq:proofhatHoptimal}
    \inf_{P\in\cP} E_P[U_T(x+\hat{H}\sint S_{T})] \geq U_0(x).
  \end{equation}
  Let $t\in\{0,\dots,T-1\}$. By the definition of $\hat{H}$, we have
  \[
    \inf_{P'\in\cP_t(\omega)}E_{P'}[U_{t+1}(\omega\otimes_t\cdot, x+\hat{H}\sint S_{t+1}(\omega,\cdot))] = U_t(\omega,x+\hat{H}\sint S_t(\omega))
  \]
  for all $\omega$ outside a $\cP$-polar set.
  Let $P\in\cP$; then $P=P_0\otimes  \cdots \otimes P_{T-1}$ for some selectors $P_s$ of $\cP_s$, $s=0,\dots,T-1$ and we conclude via Fubini's theorem that
  \begin{align*}%
    E_{P} [U_{t+1}(x+\hat{H}\sint S_{t+1})]
    & = E_{(P_0\otimes \cdots \otimes P_{t-1})(d\omega)} [ E_{P_t(\omega)}[U_{t+1}(\omega\otimes_t\cdot,  x+\hat{H}\sint S_{t+1}(\omega,\cdot))] \\
    & \geq E_{P_0\otimes \cdots \otimes P_{t-1}} [U_t(x+\hat{H}\sint S_t)]\\
    & = E_{P} [U_t(x+\hat{H}\sint S_t)].
  \end{align*}
  A repeated application of this inequality shows that $E_P[U_{T}(x+\hat{H}\sint S_{T})] \geq U_0(x)$, and since $P\in\cP$ was arbitrary, our claim~\eqref{eq:proofhatHoptimal} follows.

  To conclude that $\hat{H}$ is optimal, it remains to prove that
  \[
    U_0(x)\geq \sup_{H\in \cH_x} \inf_{P\in\cP} E_P[U(x+H \sint S_T)]=:u(x).
  \]
  To this end, we fix an arbitrary $H\in\cH_x$ and first show that
  \begin{equation}\label{eq:proofDPPindClaim}
    \inf_{P\in\cP} E_P[U_t(x+H\sint S_t)] \geq \inf_{P\in\cP} E_P[U_{t+1}(x+H\sint S_{t+1})], \quad t=0,1,\dots,T-1.
  \end{equation}
  Let $\eps>0$ and let $P^\eps_t$ be an $\eps$-optimal selector as in Lemma~\ref{le:epsOptimalP}.
  Then, for $\omega$ outside a $\cP$-polar set, Lemma~\ref{le:supOverQsame} yields that
  \begin{align*}
    &E_{P^\eps_t(\omega)} [U_{t+1}(\omega\otimes_t\cdot, x+H\sint S_{t+1}(\omega,\cdot))] - \eps \\
    &\leq (-\eps^{-1}) \vee \inf_{P\in\cP_t(\omega)} E_P[U_{t+1}(\omega\otimes_t\cdot, x+H\sint S_{t+1}(\omega,\cdot))] \\
    &\leq (-\eps^{-1}) \vee\sup_{h\in D_{t,x'(\omega)}(\omega)} \inf_{P\in\cP_t(\omega)} E_P[U_{t+1}(\omega\otimes_t\cdot, x+H\sint S_t(\omega) + h\DS_{t+1}(\omega,\cdot))]\\
    &= (-\eps^{-1}) \vee U_t(\omega, x+H\sint S_t(\omega)),
  \end{align*}
  where $x'(\omega):=x+H\sint S_t(\omega)$.
  Given $P\in\cP$, we thus have
  \begin{align*}
    E_P[(-\eps^{-1}) \vee U_t(x+H\sint S_t)]
    &\geq E_{P\otimes_t P^\eps_t} [U_{t+1}(x+H\sint S_{t+1})] - \eps \\
    &\geq \inf_{P'\in\cP}E_{P'} [U_{t+1}(x+H\sint S_{t+1})] -\eps.
  \end{align*}
  As $\eps>0$ and $P\in\cP$ were arbitrary, \eqref{eq:proofDPPindClaim} follows.
  Noting the trivial equality $U_0(x) = \inf_{P\in\cP} E_P[U_0(x+H\sint S_0)]$, a repeated application of~\eqref{eq:proofDPPindClaim} yields
  \[
    U_0(x) \geq \inf_{P\in\cP} E_P[U_1(x+H\sint S_1)] \geq \cdots \geq \inf_{P\in\cP} E_P[U_T(x+H\sint S_T)].
  \]
  As $H\in\cH_x$ was arbitrary, it follows that $U_0(x)\geq u(x)$, and in view of~\eqref{eq:proofhatHoptimal}, this shows that $\hat{H}$ is optimal.
\end{proof}

\begin{proof}[Proof of Example~\ref{ex:bddS}.]
  The assumption that $S$ is uniformly nondegenerate clearly implies that $\NA(\cP_t(\omega))$ holds for all $(t,\omega)$; in particular, Lemma~\ref{le:NAlocal} shows that $\NA(\cP)$ holds. Moreover, using that $S$ is bounded and uniformly nondegenerate, we obtain a universal constant $a$ such that $|x+ H\sint S_T|\leq a$\ for all $H\in\cH_x$.
  On the other hand, by a scaling argument similar to~\eqref{eq:scalingIneq}, the assumption that $U(\cdot,1)$ is bounded from above implies that $U(\cdot,y)\leq c_y$ for some $c_y\in\R$, for all $y>0$.
  Together, it follows that $U(x+H \sint S_T)\leq c_a$ for all $H\in\cH_x$. Using these and similar arguments, we can go through the proof of Theorem~\ref{th:multiperiodBdd} with minor modifications.
\end{proof}

\newcommand{\dummy}[1]{}

\end{document}